\documentclass{llncs}
\usepackage{amsmath}
\usepackage{amssymb}
\usepackage{mathabx}
\usepackage{url}
\usepackage{graphicx}
\usepackage{multirow}
\usepackage{color}
\usepackage[nearskip=5pt]{subfig}
\usepackage{proof}
\usepackage{hhline}
\usepackage{wrapfig}
\usepackage{tikz}
\usetikzlibrary{arrows,shapes,calc}

\definecolor{pinegreen}{rgb}{0,0.55,0.45}
\definecolor{dkpurple}{rgb}{0.25,0.25,0.5}
\definecolor{commcolour}{rgb}{0.25,0.25,0.5}

\newcommand{\wtab}[1]{
  \begin{tabular}{c}
    #1
  \end{tabular}
}

\newcommand{\pcode}[2][\codesize]{
    \fbox{
    \begin{minipage}{0.45\linewidth}
    #1
    \begin{tabbing}
    xx \= xx \= xx \= xx \= xx \= xx \= xx \= xx \= xx \= \kill
    #2
    \end{tabbing}
    \end{minipage}
    }
  }
  
\newcommand{\cfg}[5][scale=0.8]{
  \begin{tikzpicture}[yscale=-1,#1]
    \tikzstyle{terminal}=[draw,circle,minimum size=5pt,inner sep=0pt]
    \tikzstyle{stmt}=[draw,rectangle,inner sep=3pt]
    \tikzstyle{edge}=[draw,thick,-stealth]

    \foreach \pos/\id in {#2}
      \node[terminal] (\id) at \pos {};

    \foreach \pos/\id/\lbl in {#3}
      \node[stmt] (\id) at \pos {\lbl};

    \foreach \start/\dest in {#4}
      \path[edge] (\start) -- (\dest);

    #5
  \end{tikzpicture}
}

\newcommand{\ignore}[1]{}
\newcommand{\longmath}[1]{}
\newcommand{\shortmath}[1]{#1}

\newcommand{\abs}[1]{\vert#1\vert} 
\newcommand{\ebox}{\quad{\vrule height4pt width4pt depth0pt}}

\newcommand{\Va}{\ensuremath{V_{\mathcal{A}}}}
\newcommand{\Ca}{\ensuremath{C_{\mathcal{A}}}}

\DeclareMathOperator{\meet}{\sqcap}
\DeclareMathOperator{\glb}{\bigsqcap}

\DeclareMathOperator{\join}{\sqcup}
\DeclareMathOperator{\bleq}{\sqsubseteq}

\DeclareMathOperator{\unop}{\circ}
\DeclareMathOperator{\binop}{\diamond}
\DeclareMathOperator{\logop}{\bowtie}

\newcommand{\sys}[1]{\textsc{#1}}

\newcommand{\sem}[1]{\ensuremath{[\![#1]\!]}}

\DeclareCaptionType{copyrightbox}

\begin{document}

\title{A Partial-Order Approach \\ to Array Content Analysis}

\author{Graeme Gange\inst{1} \and 
        Jorge A. Navas\inst{2} \and
        Peter Schachte\inst{1} \and
        Harald S{\o}ndergaard\inst{1} \and
        Peter J. Stuckey\inst{1}}
\institute{Department of Computing and Information Systems, 
\\ The University of Melbourne, Victoria 3010, Australia
\\ \email{\{gkgange,schachte,harald,pstuckey\}@unimelb.edu.au
\\}
\and NASA Ames Research Center, Moffett Field, CA 94035, USA
\\ \email{jorge.a.navaslaserna@nasa.gov}}
\date{}

\maketitle

\begin{abstract}
We present a parametric abstract domain for array content analysis.
The method maintains invariants for contiguous regions of the array,
similar to the methods of Gopan, Reps and Sagiv, and of Halbwachs 
and P{\'e}ron.
However, it introduces a novel concept of an array content graph,
avoiding the need for an up-front factorial partitioning step. 
The resulting analysis can be used with arbitrary
numeric relational abstract domains; 
we evaluate the domain on a range of array manipulating program fragments.
\end{abstract}




\section{Introduction}
\label{sec-intro}

Most imperative programming languages offer mutable arrays.
However, owing to the indirect relation between storage and
retrieval, arrays are not particularly amenable to static analysis.
While analysis of array \emph{bounds}
(e.g,~\cite{SuzukiI77,XiP98,Bodik_ABCD}) is well studied, only
recently has there been real progress in analyzing array content.
Early approaches involved \emph{array smashing}~\cite{Blanchet_Smashing}, 
where an entire array was treated as a single symbolic variable.  
This used the concept of \emph{weak} updates: transfer functions for
array assignment that can only weaken the previous abstract state.
Weak updates generally lead to a rapid loss of
precision when different segments of an array have different properties.

A significant improvement was Gopan, Reps and Sagiv's use of
\emph{array partitioning}~\cite{Gopan_POPL05} 
to split arrays into symbolic intervals, or segments.
Partitioning is facilitated by an initial analysis of the array
index expressions to determine the relative order of the indices
they denote.
A key idea is to distinguish segments that represent single array 
cells from those that represent multiple cells. 
This permits \emph{strong} updates on the singleton segments.
The array partitioning method selects a small set of
partition variables, maintaining disjunctive information about
properties which hold over all feasible total orderings of the
partition variables.

\begin{wrapfigure}[20]{r}{0.425\textwidth}
\centerline{
  \pcode{
    $i_1$ := $0$ ; \ldots ; $i_m$ := $0$ \\
    $x$ := $\star$ \\
    \textbf{while}($i_1 < n \wedge \ldots \wedge i_m < n$) \\
    \> $p$ := $\star$ \\
    \> \textbf{if}($p < 0$) \\
    \> \> $A[i_1]$ := $x + 1$ \\
    \> \> $i_1$ := $i_1+1$ \\
    \> \textbf{else if}($p < 1$) \\
    \> \> $A[i_2]$ := $x + 2$ \\
    \> \> $i_2$ := $i_2+1$ \\
    \> \ldots \\
    \> \textbf{else} \\
    \> \> $A[i_m]$ := $x + m$ \\
    \> \> $i_m$ := $i_m+1$
  }
}
\caption{\label{fig-initrand}
  The \texttt{init\_rand}$_m$ family of program fragments.
}
\end{wrapfigure}
Halbwachs and P{\'e}ron~\cite{Halbwachs_Peron_PLDI08} extended the
approach to support relational content domains and
a limited form of quantified invariants.
The resulting method is precise, but has its own drawbacks. 
First, it requires an initial segmentation phase, where the set of 
partition variables is identified; and as this phase is purely 
syntactic, it is possible for variables to be omitted which are 
critical to the invariant. 
Second, there are exponentially many possible total orderings of the 
partition variables; if many partition variables are identified, 
the analysis may become prohibitively expensive. 
For example, on the $\texttt{init\_rand}_m$ family of programs shown in
Figure~\ref{fig-initrand}, the number of partitions at the loop head
follows the progression $[6, 30, 222, 2190, 27006]$ as $m$ increases from 1
to 5 (this is discussed in more detail in Appendix C).
Finally, the analysis does not support arbitrary manipulation of index
variables; indices may only be incremented and decremented.

Cousot \emph{et al.}~\cite{Cousot_Logozzo_POPL11} instead maintain a
single partitioning of the array, selecting a consistent totally ordered
subset from a scalar variable analysis. This does not require a
separate segmentation phase, saving considerable overhead;
however, as it considers only a single consistent ordering (and
supports only value domains) the invariants it derives are
quite weak. Consider $\texttt{init\_rand}_2$,
with $x$ fixed to $0$---so each element in $[0, n)$ will be
assigned either $1$ or $2$. In this case, the relationship between
$i_1$ and $i_2$ is not known, so we must
select either $0 \leq i_1 < n$ or $0 \leq i_2 < n$. In either case,
the desired invariant at the loop exit is lost.

An alternative approach to expressing array properties is to lift an
abstract domain to quantified 
invariants~\cite{Gulwani_Lifting_POPL08}. 
This technique is quite general but there are two major limitations. 
First, it requires from the user the specification of \emph{templates}
to describe when quantifiers are introduced.
Second, it is expensive, owing to the computation of
\emph{under-approximations}.
For example, to join the formulas $\forall U (G
\Rightarrow e)$ and $\forall U (G' \Rightarrow e')$,
we must compute an under-approximation of
$G \sqcup G'$, since $G$ and $G'$ are in negative positions,
and this is prohibitively expensive for many domains.

Dillig \emph{et al.}~\cite{Dillig_ESOP10} replace
strong and weak updates with \emph{fluid} updates.  
Their method is a points-to and value analysis,
so not relational in our sense.
It builds a points-to graph where nodes represent
abstract locations that include arrays qualified by index variables. 
Edges represent constraints on index variables that
identify which concrete elements in the source location point to which
concrete location in the target. A fluid update removes the dichotomy
between strong and weak updates by computing first a constraint $\varphi$
representing the elements that are modified by the update. Then it adds
a new points-to edge with $\varphi$ (strong update) while adding
the negation of $\varphi$ to existing edges from the source (weak update).
As $\varphi$ is an over-approximation, its negation is an
under-approximation and thus it would be unsound to add it directly to
other edges. Instead the analysis produces \emph{bracketing
  constraints} which are pairs of over- and under-approximations so that
negation can be done in a sound manner. This analysis is very
expressive, avoiding the large number of explicit partitions fixed a
priori in~\cite{Gopan_POPL05,Halbwachs_Peron_PLDI08}. However, the
method can still be very expensive since whenever an array is accessed,
the points-to edges must be modified by adding possibly disjunctive
formulas.

We propose a new approach to array content analysis.
We extend any existing scalar domain by introducing a pseudo-variable
to refer to segments of each array, selecting index expressions as
nodes in a graph, and annotating the graph edges with the properties
that hold in the segments of the arrays between those index expressions.
These \emph{array content graphs} offer greater flexibility than other
approaches, as they allow us to reason about properties that hold over
contiguous array segments without committing to a single total
ordering on index expressions,
while still taking advantage of available partial ordering information.
The result is an array content analysis which is fully
automatic, can be used with arbitrary
domains, and does not incur the up-front factorial cost
of previous methods~\cite{Gopan_POPL05,Halbwachs_Peron_PLDI08}.
In particular, it can be used for relational analyses, and accounts
for the possibility of array elements being related to array indices.

\label{sec-lang}
We base our presentation on a small control flow graph language.
\begin{quote}
  \begin{tabular}{lrcl}
    Instructions &\textsf{I} & ~$\rightarrow$~ & $v_1 = \mathit{constant}$ 
                 $|$ $v_1 = \unop v_2$   $|$ $v_1 = v_2 \binop v_3$  $|$ \textsf{A} \\
    Array assignments~ &\textsf{A} & ~$\rightarrow$~ & $v_1 = arr[v_2]$ $|$ $arr[v_1] = v_2$ \\
    Jumps & \textsf{J} & ~$\rightarrow$~ &  $\mathit{if}$ $(v_1 \logop v_2)$
    $\mathit{label}_1$ $\mathit{label}_2$  $|$ $\mathit{br}$ $\mathit{label}$ $|$ $error$ $|$ $end$ \\
    Blocks & \textsf{B} & ~$\rightarrow$~ & $\mathit{label}$ : \textsf{I}* \textsf{J} \\
    Programs & \textsf{P} & ~$\rightarrow$~ & $\textsf{B}^+$ \\
  \end{tabular}
\end{quote}
Each basic block is a (possibly empty) sequence of
instructions, ending in a (possibly conditional) jump.
Arithmetic unary and binary operators are denoted by
$\unop$ and $\binop$ respectively, and comparison operators by
$\logop$.  We assume that there is a fixed set of arrays $\{A_1,
\ldots, A_k\}$, which have global scope (and do not overlap in memory).
The semantics is conventional and not discussed here.
Figure~\ref{fig-copy} shows an example program in diagrammatic form.

Our analysis assumes an abstract domain
${\cal L} = \langle L, \sqsubseteq, \bot, \top, \sqcup, \sqcap \rangle$ 
for analysis over the array-free fragment of the language 
(obtained by leaving out \textsf{A})
(scalar analysis). 
We use this parametric domain to construct the array content analysis.

The remainder of this paper is structured as follows.
Section~\ref{sec-domain}
introduces the method and its underlying ideas.
Section~\ref{sec-efficiency} discusses computational details and
efficiency.
The method has been evaluated experimentally; 
Section~\ref{sec-exper} gives a report and
Section~\ref{sec-future} concludes, suggesting further work.

\section{A Graph-Based Array Content Domain}
\label{sec-domain}

We let $V$ and $\mathcal{A}$ be sets of scalar and array variables,
respectively.
%
A state in the concrete domain is a pair $\langle \sigma, \rho \rangle$,
where $\sigma:V \mapsto \mathbb{Z}$ maps
scalar variables to integer values, 
and $\rho:\mathcal{A} \rightarrow \mathbb{Z} \rightarrow \mathbb{Z}$
maps array cells to values.\footnote{For simplicity we assume arrays
elements are integers.  
The extension to arbitrary types (that may include integers) 
is not difficult, as the complexity of array elements acting as
indices is present in what we consider.}

\begin{wrapfigure}[19]{2}{0.4\textwidth}
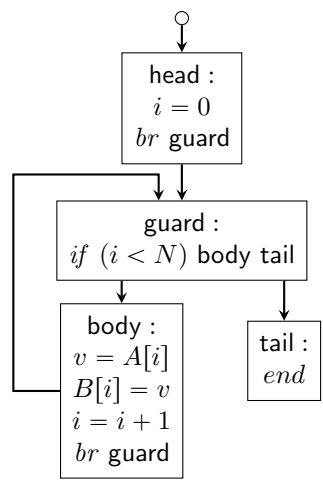

  \vspace*{-1em}
  \centerline{
    \cfg
      {{(0,-0.7)/p0}}
      {{(0, 0.8)/head/\wtab{$\mathsf{head:}$ \\
                          $i = 0$ \\
                          $\mathit{br}$ $\mathsf{guard}$}},
       {(0, 3)/guard/\wtab{$\mathsf{guard:}$ \\
                          $\mathit{if}$ $(i < N)$ $\mathsf{body}$ $\mathsf{tail}$}},
       {(-1, 5.5)/body/\wtab{$\mathsf{body:}$ \\
                          $v = A[i]$ \\
                          $B[i] = v$ \\
                          $i = i+1$ \\
                          $\mathit{br}$ $\mathsf{guard}$}},
       {(1.70, 5)/tail/\wtab{$\mathsf{tail:}$ \\
                          $end$}}}
      {{p0/head},{head/guard}}
      {
        \path[edge] (body) -| ++(-1.8,-3.6) -| (guard.120);   
        \path[edge] (guard.south-|body) -- (body);
        \path[edge] (guard.south-|tail) -- (tail);
      }
  }
  \caption{\label{fig-copy}
    Copying array $A$ to $B$}
\end{wrapfigure}
Let $\Va$ denote the set of variables,
and $\Ca$ the set of constants, which may act as segment bounds.
We use an extended set of variables 
$V' = \Va \cup \Ca \cup \{ v{^+} \mid v \in \Va \cup \Ca \}$.
The pseudo-variable $v{^+}$ represents the value $v+1$.
This allows us to talk about properties that apply to single array
cells.
The constant vertices are often needed; 
because array properties often hold over ranges bounded by a constant 
on one side, and array processing code often initializes a fixed set 
of indices before processing the rest of the array,
it is not sufficient to consider only those variables and constants 
used directly as indices.
Consider the \texttt{copy} program shown in Figure~\ref{fig-copy}.
In this case neither $0$ nor $N$ is ever used \emph{directly} as an index,
but both define boundary conditions of $i$, and are needed for the 
invariant of interest.
In practice, however, $\Va$ is often considerably smaller than $V$.

We wish to relate the value of elements in an array segment
to variables in the scalar domain.
A state in the abstract domain has the form $\langle\varphi,\Psi\rangle$,
where $\varphi$ expresses scalar properties and $\Psi$ expresses array
content properties.
For each array $A \in \mathcal{A}$,
we allocate a corresponding variable $a$,
a \emph{segment} variable, which occurs only in $\Psi$, 
never in $\varphi$.
Relations between scalar and segment variables are captured in $\Psi$.
We use $U = \{a, b, \ldots \}$ to denote the set of segment variables.
Sometimes, we may wish to relate the values in an array segment
to the corresponding index, for example, to prove that $A[i] = x+i$
across a segment.
To support this, we introduce a variable $idx$ to represent the index
of a given read. We use $U_{I} = \{idx\} \cup U$ to denote the 
augmented set of segment variables.

The analyses of the scalar domain and array contents are based on the
same lattice $L$. 
%
We represent the program state as a pair of the scalar 
properties $\varphi$ and the $\abs{V'}\times\abs{V'}$ matrix $\Psi$ of array 
properties such that $\psi_{ij}$ denotes the properties which hold for 
all indices in the interval $[i, j)$. 
In a slight abuse of notation, we use $\psi(\ell)$ to denote the formula 
$\psi$ with each symbolic array variable $a$ replaced
with the corresponding array element $A[\ell]$. 
That is:
\[
  \psi(\ell) \equiv \psi \{ a \mapsto A[\ell] \mid A \in \mathcal{A} \}
\]

\noindent
Then the reading of an edge $(i, \psi_{ij}, j)$ is:
$\forall~i \leq \ell < j ~.~ \psi_{ij}(\ell)$.
Given some numeric abstract domain $\mathcal{L}$, a set of arrays $A$ 
and scalar variables $V$, the array content domain 
$\mathcal{C}_{\mathcal{L}}(A, V')$ is a pair $\langle \varphi, \Psi \rangle$,
where $\varphi$ is a value in $\mathcal{L}$, and 
$\Psi$ is a $\abs{V'}\times\abs{V'}$ matrix of $\mathcal{L}$-values.
Assume we have a function $\mathsf{eval}_{\mathcal{A}}$ which 
constructs a new state which treats
the $\ell^{th}$ element of each array as a scalar variable:

\[
  \mathsf{eval}_{\mathcal{A}}(\sigma,\rho,\ell) = \sigma \cup \{ a \mapsto \rho(A)(\ell) \mid A \in \mathcal{A}\}
\]
The concretization function $\gamma$ is then defined on the components
in terms of the concretization function $\gamma_{\cal L}$ of the 
scalar domain ${\cal L}$:
  $$\gamma(\langle \varphi, \Psi \rangle) = \{ \langle \sigma, \rho \rangle \mid \sigma \in \gamma_{ \mathcal{L}}(\varphi) \} 
 \cap \gamma_A(\Psi)$$
  $$\gamma_A(\Psi) = \bigcap_{i, j \in V'} \{ \langle \sigma, \rho \rangle \mid
    \forall_{\sigma(i) \leq \ell < \sigma(j)} \mathsf{eval}_{\mathcal{A}}(\sigma, \rho, \ell) \in \gamma_{\mathcal{L}}(\psi_{ij}) \}$$
%
Notice that there are no constraints on $\rho$ in the first equation;
$\rho$ can be any array variable assignment of type
$\mathcal{A} \rightarrow \mathbb{Z} \rightarrow \mathbb{Z}$.
As any value in the content domain 
$\mathcal{C}_{\mathcal{L}}(A,V')$ is the Cartesian product of a 
fixed set of elements of $\mathcal{L}$, $\mathcal{C}_{\mathcal{L}}(A,V')$
also forms a lattice, and possesses all the corresponding fixed point 
properties.

For each edge $\psi_{ij}$, we can assume the corresponding 
interval is non-empty; that is, $\sem{i < j} \in \psi_{ij}$.
Note that the edge from $i$ to $i{^+}$ has no such constraint, since $i < i+1$ is always true.
Since the interval $[i+1, i)$ is clearly empty, $\psi_{i{^+}i} = \bot$.

\begin{figure}[t]
  \centerline{
$\varphi = 0 \leq i < N$ ~~~~ $\Psi = $
\begin{minipage}[c]{0.3\textwidth}
$$
\begin{array}{l|cccc}
  & 0 & i & i^+ & N \\
\hline
0 & \bot & a = b & \top & \top \\
i & \bot & \bot & a = v & \top \\
i^+ & \bot & \bot & \bot & \top \\
N   & \bot & \bot & \bot & \bot 
\end{array}
$$
\end{minipage}
~~~
\begin{minipage}{0.3\textwidth}
\begin{tabular}{c}
    \begin{tikzpicture}[xscale=1.2]
      \tikzstyle{vertex}=[draw,circle,minimum size=5pt,inner sep=0pt]
      \foreach \pos/\id/\lbl in {{(0, 0)/0/0}, {(1, 0)/i/i}, {(1.9, 0)/ip/i{^+}}, {(2.9, 0)/N/N}}
      {
        \node[vertex] (\id) at \pos {};
        \path (\id) ++(0, -0.3) node {$\lbl$};
      }
      \foreach \src/\dest/\label in {{0/i/$a = b$},{i/ip/$a = v$},{ip/N/$\top$}}
        \path[draw, thick, -stealth] (\src) -- node[above] {\label} (\dest);
      \path[draw, thick, -stealth] (0) .. controls (0.3, -0.75) and (1.6, -0.75) .. node[below] {$\top$} (ip);
      \path[draw, thick, -stealth] (i) .. controls (1.3, 0.75) and (2.6, 0.75) .. node[above] {$\top$} (N);
      \path[draw, thick, -stealth] (0) .. controls (0.3, 1.5) and (2.6, 1.5) .. node[above] {$\top$} (N);
    \end{tikzpicture}
\\
    \begin{tikzpicture}[xscale=1.2]
      \tikzstyle{vertex}=[draw,circle,minimum size=5pt,inner sep=0pt]
      \foreach \pos/\id/\lbl in {{(0, 0)/0/0}, {(1, 0)/i/i}, {(1.9, 0)/ip/i{^+}}, {(2.9, 0)/N/N}}
      {
        \node[vertex] (\id) at \pos {};
        \path (\id) ++(0, -0.3) node {$\lbl$};
      }
      \foreach \src/\dest/\label in {{0/i/$a = b$},{i/ip/$a = v$},{ip/N/}}
        \path[draw, thick, -stealth] (\src) -- node[above] {\label} (\dest);
    \end{tikzpicture}
  \end{tabular}
\end{minipage}
}
\caption{\label{fig-copystate}
Array content graph after assignment $v = A[i]$ in Figure~\ref{fig-copy}.
Vertices and matrix entries corresponding to
$0^{+}$ and $N^{+}$ are omitted.
}
\end{figure}
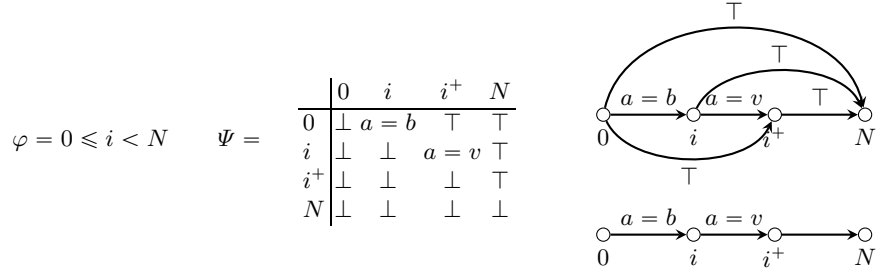

Given the state representation, we can take the join of two 
abstract states by a piecewise application of the join from ${\cal L}$:
  $$\langle \varphi^1, \Psi^1 \rangle \join \langle \varphi^2, \Psi^2 \rangle = 
    \left<
    \varphi^1 \join \varphi^2,
      \left[
      \begin{array}{ccc}
        \psi^1_{11} \sqcup \psi^2_{11} & \ldots & \psi^1_{1n} \sqcup \psi^2_{1n} \\
        \vdots & \ddots & \vdots \\
        \psi^1_{n1} \sqcup \psi^2_{n1} & \ldots & \psi^1_{nn} \sqcup \psi^2_{nn} \\
      \end{array} 
      \right]
    \right>
  $$
We can compute the meet $\meet$ analogously. 
To see how the analysis works, consider again the program in 
Figure~\ref{fig-copy}.
Figure~\ref{fig-copystate} shows the abstract state
immediately after executing $v = A[i]$. 
The array content information is given by the matrix $\Psi$ of array
properties.
The \emph{array content graph} shown upper right is really
just a way of visualizing the matrix.
Note that infeasible edges, those labelled $\bot$, are omitted. 
In fact, we shall usually show only the ``transitive reduction''
of the array content graph, so that an edge $ik$ whose value is
given by $\psi_{ik} = \glb_j (\psi_{ij} \join \psi_{jk})$ is omitted. 
We depict edges representing $\top$ without a label.
This leads to the shorthand graph in Figure~\ref{fig-copystate}'s 
lower right.

\subsection{Normalization of Abstract States}
\label{sec-normal}

Given a set $C$ of constraints of the form 
$e \in [i, j) \Rightarrow \psi_{ij}$, we wish to normalize the state 
by computing the strongest consequences of $C$, still in that form.

The critical observation is this: 
  $\forall x, y~.~x < y \models \forall z.(x < z) \vee (z < y)$.
That is, any property that holds over both $[x, z)$ and $[z, y)$
must also hold over the range $[x, y)$.
To compute the strongest consequences, then, we must compute the greatest
fixed point of a rewrite system derived from the set of inequalities:
$\forall i, j, k~.~ \psi_{ij} \sqsubseteq \psi_{ik} \join \psi_{kj}$.

\begin{wrapfigure}[22]{r}{0.54\textwidth}
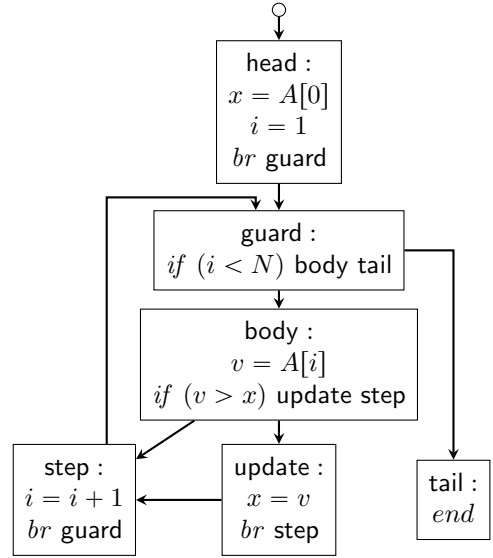

  \vspace*{-6ex}
  \centerline{
    \cfg
      {{(0,-0.9)/p0}}
      {{(0, 0.8)/head/\wtab{$\mathsf{head:}$ \\
                          $x = A[0]$ \\
                          $i = 1$ \\
                          $\mathit{br}$ $\mathsf{guard}$}},
       {(0, 3.1)/guard/\wtab{$\mathsf{guard:}$ \\
                           $\mathit{if}$ $(i < N)$ $\mathsf{body}$ $\mathsf{tail}$}},
       {(0, 5.00)/body/\wtab{$\mathsf{body:}$ \\
                          $v = A[i]$ \\
                          $\mathit{if}$ $(v > x)$ $\mathsf{update}$ $\mathsf{step}$}},
       {(0, 7.25)/update/\wtab{$\mathsf{update:}$ \\
                          $x = v$ \\
                          $\mathit{br}$ $\mathsf{step}$}},
       {(-3.4, 7.25)/step/\wtab{$\mathsf{step:}$ \\
                          $i = i+1$ \\
                          $\mathit{br}$ $\mathsf{guard}$}},
       {(2.9, 7.25)/tail/\wtab{$\mathsf{tail:}$ \\
                          $end$}}}
      {{p0/head},{head/guard},{guard/body},{body/update},{body/step},{update/step}}
      {
        \path[edge] (step.60) -- ++(0,-4.1) -| (guard.120);   
        \path[edge] (guard.east) -| (tail.north);   
      }
  }
  \caption{\label{fig-arraymax}
    Find maximum value in array $A$
  }
\end{wrapfigure}
It is tempting to try to compute the fixed point using the obvious 
rewrite system:
  $$\forall i, j, k~.~ \psi_{ij} = \psi_{ij} \meet (\psi_{ik} \join \psi_{kj})$$
But, as we shall see, this formulation does
not guarantee termination for all useful domains.

Mohri~\cite{mohri-semiring} describes the \emph{algebraic shortest path} 
problem, where the operations $(\meet, \join)$ form a semiring in 
the domain of edge weights.
This is very close to what we need, as every distributive lattice forms 
a bounded semiring. 
But ``numeric'' domains used in static analysis generally fail to be 
distributive, so we cannot use Mohri's framework directly.

\begin{figure}
  \begin{center}
  \begin{tabular}{ccc}
    (1) & 
    $0 \leq i \leq N$ & 
    \begin{tikzpicture}[xscale=2.5, baseline=(i.base)]
      \tikzstyle{vertex}=[draw,circle,minimum size=5pt,inner sep=0pt]
      \foreach \pos/\id/\lbl in {{(0, 0)/0/0}, {(0.8, 0)/i/i}, {(1.6, 0)/ip/i{^+}}, {(2.4, 0)/N/N}}
      {
        \node[vertex] (\id) at \pos {};
        \path (\id) ++(0, -0.3) node {$\lbl$};
      }
      \foreach \src/\dest/\label in {{0/i/$a \leq x$},{i/ip/$a = v$},{ip/N/}}
        \path[draw, thick, -stealth] (\src) -- node[above] {\label} (\dest);
    \end{tikzpicture}
    \\
    (2a) &
    \wtab{
      $0 \leq i \leq N$ \\ 
      $v \leq x$
    } &
    \begin{tikzpicture}[xscale=2.5, baseline=(i.base)]
      \tikzstyle{vertex}=[draw,circle,minimum size=5pt,inner sep=0pt]
      \foreach \pos/\id/\lbl in {{(0, 0)/0/0}, {(0.8, 0)/i/i}, {(1.6, 0)/ip/i{^+}}, {(2.4, 0)/N/N}}
      {
        \node[vertex] (\id) at \pos {};
        \path (\id) ++(0, -0.3) node {$\lbl$};
      }
      \foreach \src/\dest/\label in {{0/i/$a \leq x$},{i/ip/\wtab{$a = v$, \color{pinegreen} $\underline{a \leq x}$}},{ip/N/}}
        \path[draw, thick, -stealth] (\src) -- node[above] {\label} (\dest);
    \end{tikzpicture}
    \\
    (2b) &
    \wtab{
      $0 \leq i \leq N$ \\
      $x < v$
    } &
    \begin{tikzpicture}[xscale=2.5, baseline=(i.base)]
      \tikzstyle{vertex}=[draw,circle,minimum size=5pt,inner sep=0pt]
      \foreach \pos/\id/\lbl in {{(0, 0)/0/0}, {(0.8, 0)/i/i}, {(1.6, 0)/ip/i{^+}}, {(2.4, 0)/N/N}}
      {
        \node[vertex] (\id) at \pos {};
        \path (\id) ++(0, -0.3) node {$\lbl$};
      }
      \foreach \src/\dest/\label in {{0/i/\wtab{$a \leq x$, $\color{pinegreen} \underline{a < v}$}},{i/ip/$a = v$},{ip/N/}}
        \path[draw, thick, -stealth] (\src) -- node[above] {\label} (\dest);
    \end{tikzpicture}
    \\
    (2b$'$) &
    \wtab{
      $0 \leq i \leq N$ \\
      $x = v$
    } &
    \begin{tikzpicture}[xscale=2.5, baseline=(i.base)]
      \tikzstyle{vertex}=[draw,circle,minimum size=5pt,inner sep=0pt]
      \foreach \pos/\id/\lbl in {{(0, 0)/0/0}, {(0.8, 0)/i/i}, {(1.6, 0)/ip/i{^+}}, {(2.4, 0)/N/N}}
      {
        \node[vertex] (\id) at \pos {};
        \path (\id) ++(0, -0.3) node {$\lbl$};
      }
      \foreach \src/\dest/\label in {{0/i/\wtab{$\color{pinegreen} \underline{a < v}$, {\color{pinegreen} $\underline{a < x}$}}},
                                     {i/ip/\wtab{$a = v$, {\color{pinegreen} $\underline{a = x}$}}},
                                     {ip/N/}}
        \path[draw, thick, -stealth] (\src) -- node[above] {\label} (\dest);
    \end{tikzpicture}

    \\
    (3) &
    \wtab{
      $0 \leq i \leq N$
    } & 
    \begin{tikzpicture}[xscale=2.5, baseline=(i.base)]
      \tikzstyle{vertex}=[draw,circle,minimum size=5pt,inner sep=0pt]
      \foreach \pos/\id/\lbl in {{(0, 0)/0/0}, {(0.8, 0)/i/i}, {(1.6, 0)/ip/i{^+}}, {(2.4, 0)/N/N}}
      {
        \node[vertex] (\id) at \pos {};
        \path (\id) ++(0, -0.3) node {$\lbl$};
      }
      \foreach \src/\dest/\label in {{0/i/\wtab{\color{pinegreen} $\underline{a \leq x}$}},
                                     {i/ip/\wtab{$a = v$, \color{pinegreen} $\underline{a \leq x}$}},
                                     {ip/N/}}
        \path[draw, thick, -stealth] (\src) -- node[above] {\label} (\dest);
    \end{tikzpicture}
    \\
    \hline
    (2b$^\dagger$) &
    \wtab{
      $0 \leq i \leq N$ \\
      $\color{pinegreen} \underline{x = v}$
    } &
    \begin{tikzpicture}[xscale=2.5, baseline=(i.base)]
      \tikzstyle{vertex}=[draw,circle,minimum size=5pt,inner sep=0pt]
      \foreach \pos/\id/\lbl in {{(0, 0)/0/0}, {(0.8, 0)/i/i}, {(1.6, 0)/ip/i{^+}}, {(2.4, 0)/N/N}}
      {
        \node[vertex] (\id) at \pos {};
        \path (\id) ++(0, -0.3) node {$\lbl$};
      }
      \foreach \src/\dest/\label in {{0/i/\wtab{$\color{pinegreen} \underline{a < x}$, $a < v$}},{i/ip/\wtab{$a = v$, $a = x$}},{ip/N/}}
        \path[draw, thick, -stealth] (\src) -- node[above] {\label} (\dest);
    \end{tikzpicture}

  \end{tabular}
  \end{center}
  \caption{
    \label{fig-maxstate}
    Array maximum example 
  }
\end{figure}
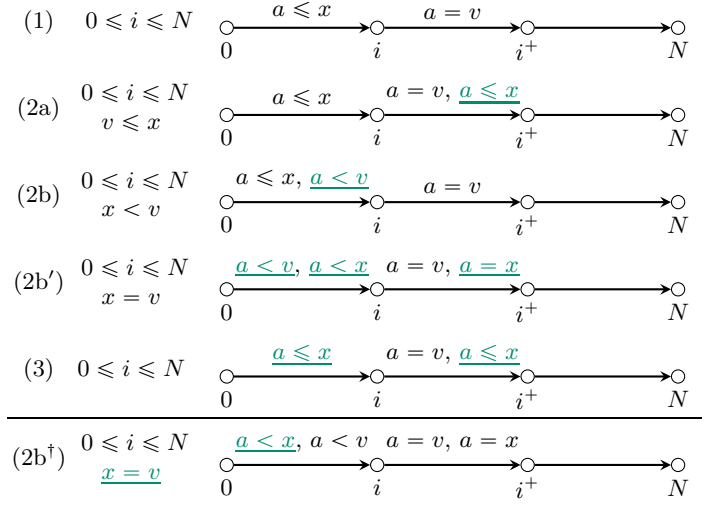

\begin{example}
\label{ex-channel}
Consider the array-maximum program from Figure~\ref{fig-arraymax}.
Figure~\ref{fig-maxstate}(1) shows the program state just 
after $v = A[i]$ is executed
(scalar constraints on the left and array constraints on the right).
On the branch with $v \leq x$, we simply add the constraint to the
scalar domain, resulting in (2a). 
If $v > x$, we add the constraint to the scalar domain ($2b$), 
then update $x$ with $v$, resulting in ($2b'$).
Observe that, in both cases, we can only discover the relationship 
between $a$ and $x$ indirectly via $v$. 
If we do not push scalar relations into the edge properties, the
underlined invariants are lost.
The final result is shown at (3).

If the statement $x$ = $v$ in $\mathsf{update}$ was instead replaced by
$x$ = $A[i]$, we would first have to lift the invariant $x = v$ from
the singleton property $\psi_{ii{^+}}$ to the scalar domain. 
We then push this property
into the segment $\psi_{0i}$, allowing us to derive $a < x$.
The state in Figure~\ref{fig-maxstate}($2b^\dagger$) shows the program 
state if $x = v$ were to be replaced with $x = A[i]$.
\hfill \ebox
\end{example}
This illustrates that it is not sufficient to simply compute the transitive
closure of $\Psi$; we must also lift properties from $\Psi$
out to the scalar domain.
For a fully reduced state $\langle \varphi, \Psi \rangle$, the 
following properties must be satisfied for all $i, j, k$:
\begin{enumerate}
  \item \label{rule_trans}
    The graph of segments must be internally consistent%
\longmath{
    $$\psi_{ij} \bleq \psi_{ik} \join \psi_{kj}$$
}%
\shortmath{:
    $\psi_{ij} \bleq \psi_{ik} \join \psi_{kj}$}
  \item \label{rule_forward}
    Segment properties are consistent with the scalar domain%
\longmath{
    $$\psi_{ij} \bleq \varphi \meet \sem{i < j}$$}%
\shortmath{: $\psi_{ij} \bleq \varphi \meet \sem{i < j}$}
  \item \label{rule_backward}
    For each non-empty segment, the scalar domain must be consistent
    with the scalar properties of that segment%
\longmath{
    $$\varphi \models i < j \Rightarrow \varphi \bleq \exists U_{I}~.~\psi_{ij}$$
}%
\shortmath{:
$\varphi \models i < j \Rightarrow \varphi \bleq \exists U_{I}~.~\psi_{ij}$}
\end{enumerate}
Notice that we only propagate constraints to the scalar component from 
segments that are known to be non-empty. 
If we tried to propagate information from all segments,
we would incorrectly derive $\bot$ as soon as any segment was 
determined to be empty.
One solution is to simply apply the three rules until a fixed point 
is reached.
This is guaranteed to compute the fully reduced state.
However, 
while this direct construction is conceptually clean, 
it suffers from some pragmatic issues relating to 
both termination and efficiency, as we shall see.

\subsection{Termination}
\label{sec:termination}
The normalization process is not guaranteed to 
terminate for arbitrary lattices.

\begin{example}
\label{ex-diverge}
Assume the analysis uses convex polyhedra.
Consider the state in Figure~\ref{fig:converge}(a).
Any fixed point will satisfy the properties
$A \bleq B \join C$ and $B \bleq A \join C$.
Let $A$, $B$ and $C$ be the gray regions shown in
Figure~\ref{fig:converge}(b)---the intention is that $A$ shares
a line segment with $C$, as does $B$.
Assume we start by exploiting $B \bleq A \join C$. 
We compute $A \join C$, yielding the polygon given by the topmost 
dashed line.
This allows us to trim the top portion of $B$. 
We then compute $B \join C$, and trim the top-left region of $A$.
However, now $A$ has changed, so we re-compute $A \join C$ and 
again reduce $B$.
This process asymptotically approaches the greatest 
fixed point $A = A \sqcap C,~B = B \sqcap C$.
\hfill \ebox
\end{example}
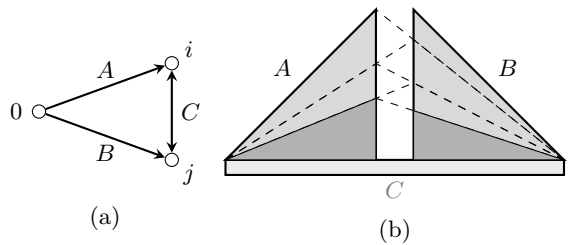
\begin{figure}[t]
\centerline{
\subfloat[]{
  \begin{tikzpicture}[baseline=(j.base)]
    \tikzstyle{vertex}=[draw,circle,minimum size=5pt,inner sep=0pt]
    \foreach \pos/\id/\lbl in {{180/0/0}, {40/i/i}, {-40/j/j}}
    {
      \node[vertex] (\id) at (\pos:1) {};
      \node at (\pos:1.3) {$\lbl$};
    }
    \path[draw, thick, -stealth] (0) -- node[above] {$A$} (i);
    \path[draw, thick, -stealth] (0) -- node[below] {$B$} (j);
    \path[draw, thick, stealth-stealth] (i) -- node[right] {$C$} (j);
  \end{tikzpicture}
}
\subfloat[]{
  \begin{tikzpicture}[baseline=(ap.base)]
    \coordinate (ax) at (-2, 0);
    \coordinate (bx) at (2.5, 0);
    \coordinate (ap) at (0, 0);
    \coordinate (bp) at (0.5, 0);
    \draw[black, thick, fill=black!15] (-2, 0) -- node[above left] {$A$} (0, 2) coordinate (a0) -- (0, 0) -- cycle;
    \draw[black, thick, fill=black!15] (2.5, 0) -- node[above right] {$B$} (0.5, 2) coordinate (b0) -- (0.5, 0) -- cycle;
    \draw[black, thick, fill=black!15, fill opacity=0.5] (-2, 0) rectangle node[below=2] {$C$} (2.5, -0.2);
    \draw[black, dashed] (bx) -- (intersection cs: first line={(bx) -- (a0)}, second line={(bp) -- (b0)}) coordinate (b1) -- (a0) -- (ax);
    \draw[black, dashed] (ax) -- (intersection cs: first line={(ax) -- (b1)}, second line={(ap) -- (a0)}) coordinate (a1) -- (b1) -- (bx);
    \draw[black, dashed] (bx) -- (intersection cs: first line={(bx) -- (a1)}, second line={(bp) -- (b1)}) coordinate (b2) -- (a1) -- (ax);
    \draw[black, dashed] (ax) -- (intersection cs: first line={(ax) -- (b2)}, second line={(ap) -- (a1)}) coordinate (a2) -- (b2) -- (bx);
    \draw[black, dashed] (bx) -- (intersection cs: first line={(bx) -- (a2)}, second line={(bp) -- (b2)}) coordinate (b3) -- (a2) -- (ax);
    \draw[fill=black!30] (ax) -- (a2) -- (ap) -- cycle;
    \draw[fill=black!30] (bx) -- (b3) -- (bp) -- cycle;
  \end{tikzpicture}
}
}
\caption{\label{fig:converge}
  At a fixed point, we have $A \bleq B \join C$ and $B \bleq A \join C$.  
  The regions $A$ and $B$ will be progressively reduced, indefinitely.
}
\end{figure}
If we modify the equations in
Example~\ref{ex-diverge} slightly, 
to $A \bleq (A \meet B) \join (A \meet C)$
(still a valid approximation of the concrete state),
convergence is immediate.

The fixed point process will clearly terminate for the 
interval domain, as the possible interval
end-points are drawn from the initial set.
We can also show that it is guaranteed to terminate for 
both octagons~\cite{Mine_Octagons_HOSC06} and convex 
polyhedra~\cite{Cousot_Halbwachs_POPL78}; proofs are given
in Appendix A.
Unfortunately, we do not yet have a more general characterisation of 
the lattices for which termination is (or is not) guaranteed.

\subsection{Abstract Transfer Functions}
In this section, we describe the abstract transfer functions necessary to perform array content analysis on
the language described in Section~\ref{sec-lang}.

\subsubsection{Variable assignment.}
The effect of a scalar assignment $x = f$ on an abstract state follows the behaviour of the underlying domain $L$.
We first project out the previous value of $x$ (assuming $x$ does not occur in $f$) then introduce the new constraint
into the scalar domain.
However, when we project $x$ from our scalar domain, we must also update all incoming and outgoing edges of $x$ (and $x{^+}$).
This becomes:
\longmath{
  $$\langle \varphi, \Psi \rangle \sem{\mathtt{x = f}} = \langle \varphi \sem{\mathtt{x = f}}, \Psi' \rangle$$
  \noindent
}
\shortmath{$\langle \varphi, \Psi \rangle \sem{\mathtt{x = f}} = \langle \varphi \sem{\mathtt{x = f}}, \Psi' \rangle$}
  where $\Psi'$ is given by, for all $p, q$:
  $$\psi'_{pq} = \left\{
      \begin{array}{cl}
        \top &~\mathbf{if}~p \in \{x, x{^+}\} \vee q \in \{x, x{^+}\} \\
        \exists x~.~\psi_{pq}~~ &~\mathbf{otherwise}
      \end{array}
    \right.
  $$
The assumption that $f$ is free of $x$ is not always well founded; 
however, we can always transform the program so that it \emph{is} the case:
\longmath{
\begin{equation*}
  \mathtt{x = f} \Rightarrow \left. \begin{array}{l}
                  \mathtt{x_t = x} \\
                  \mathtt{x = f} [\mathtt{x} / \mathtt{x_t} ]
               \end{array} \right.
\end{equation*}}%
\shortmath{$\mathtt{x = f} \Rightarrow \mathtt{x_t = x}; ~ \mathtt{x = f} [\mathtt{x} / \mathtt{x_t}]$.}
\noindent

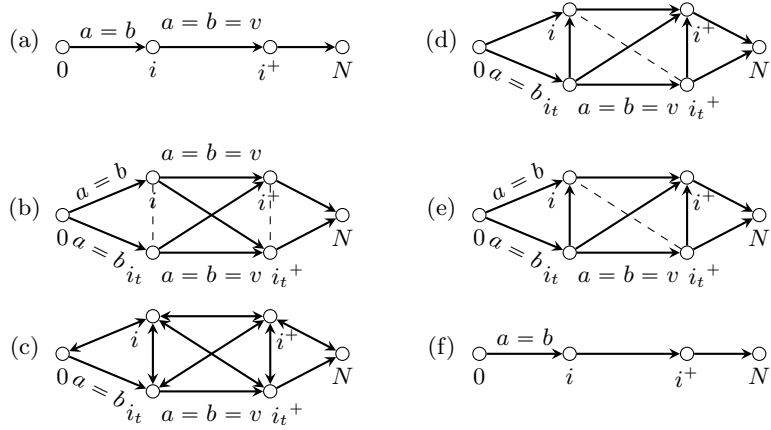
\begin{figure}[t]
  \begin{center}
  \begin{tabular}{ccccc}
    (a) &
    \begin{tikzpicture}[xscale=1.2,baseline=(0.base)]
      \tikzstyle{vertex}=[draw,circle,minimum size=5pt,inner sep=0pt]
      \foreach \pos/\id/\lbl in {{(0, 0)/0/0}, {(1, 0)/i/i}, {(2.3, 0)/ip/i{^+}}, {(3.1, 0)/N/N}}
      {
        \node[vertex] (\id) at \pos {};
        \path (\id) ++(0, -0.3) node {$\lbl$};
      }
      \foreach \src/\dest/\label in {{0/i/$a = b$},{i/ip/\wtab{$a = b = v$}},{ip/N/}}
        \path[draw, thick, -stealth] (\src) -- node[above] {\label} (\dest);
    \end{tikzpicture}
    & \qquad \qquad &
    (d) &
    \begin{tikzpicture}[xscale=1.2,baseline=(0.base)]
      \tikzstyle{vertex}=[draw,circle,minimum size=5pt,inner sep=0pt]
      \foreach \pos/\id/\lbl in {{(0, 0)/0/0}, {(3.1, 0)/N/N}}
      {
        \node[vertex] (\id) at \pos {};
        \path (\id) ++(0, -0.3) node {$\lbl$};
      }
      \node[vertex] (i) at (1, 0.5) {};
      \path (i) ++(-0.2, -0.3) node {$i$};
      \node[vertex] (ip) at (2.3, 0.5) {};
      \path (ip) ++(0.2, -0.3) node {$i{^+}$};
      \node[vertex] (j) at (1, -0.5) {};
      \path (j) ++(-0.2, -0.3) node {$i_t$};
      \node[vertex] (jp) at (2.3, -0.5) {};
      \path (jp) ++(0.2, -0.3) node {$i_t{^+}$};

      \foreach \src/\dest/\label in {{0/i/},{i/ip/},{ip/N/}}
        \path[draw, thick, -stealth] (\src) -- node[above,sloped] {\label} (\dest);
      \foreach \src/\dest/\label in {{0/j/$a = b$},{j/jp/\wtab{$a = b = v$}},{jp/N/}}
        \path[draw, thick, -stealth] (\src) -- node[below,sloped] {\label} (\dest);
      \foreach \src/\dest in {{jp/i}}
        \path[draw, dashed] (\src) -- (\dest);
      \foreach \src/\dest in {{j/i}, {j/ip}, {jp/ip}}
        \path[draw, thick, -stealth] (\src) -- (\dest);
    \end{tikzpicture} \\
    (b) &
    \begin{tikzpicture}[xscale=1.2,baseline=(0.base)]
      \tikzstyle{vertex}=[draw,circle,minimum size=5pt,inner sep=0pt]
      \foreach \pos/\id/\lbl in {{(0, 0)/0/0}, {(1, 0.5)/i/i}, {(2.3, 0.5)/ip/i{^+}},
                                                {(3.1, 0)/N/N}}
      {
        \node[vertex] (\id) at \pos {};
        \path (\id) ++(0, -0.3) node {$\lbl$};
      }
      \node[vertex] (j) at (1, -0.5) {};
      \path (j) ++(-0.2, -0.3) node {$i_t$};
      \node[vertex] (jp) at (2.3, -0.5) {};
      \path (jp) ++(0.2, -0.3) node {$i_t{^+}$};
      \foreach \src/\dest/\label in {{0/i/$a = b$},{i/ip/\wtab{$a = b = v$}},{ip/N/}}
        \path[draw, thick, -stealth] (\src) -- node[above,sloped] {\label} (\dest);
      \foreach \src/\dest/\label in {{0/j/$a = b$},{j/jp/\wtab{$a = b = v$}},{jp/N/}}
        \path[draw, thick, -stealth] (\src) -- node[below,sloped] {\label} (\dest);
      \path[draw, thick, -stealth] (j) -- (ip);
      \path[draw, thick, -stealth] (i) -- (jp);
      \path[draw, dashed] (i) -- (j);
      \path[draw, dashed] (ip) -- (jp);
    \end{tikzpicture} 
    & \qquad &
    (e) &
    \begin{tikzpicture}[xscale=1.2,baseline=(0.base)]
      \tikzstyle{vertex}=[draw,circle,minimum size=5pt,inner sep=0pt]
      \foreach \pos/\id/\lbl in {{(0, 0)/0/0}, {(3.1, 0)/N/N}}
      {
        \node[vertex] (\id) at \pos {};
        \path (\id) ++(0, -0.3) node {$\lbl$};
      }
      \node[vertex] (i) at (1, 0.5) {};
      \path (i) ++(-0.2, -0.3) node {$i$};
      \node[vertex] (ip) at (2.3, 0.5) {};
      \path (ip) ++(0.2, -0.3) node {$i{^+}$};
      \node[vertex] (j) at (1, -0.5) {};
      \path (j) ++(-0.2, -0.3) node {$i_t$};
      \node[vertex] (jp) at (2.3, -0.5) {};
      \path (jp) ++(0.2, -0.3) node {$i_t{^+}$};

      \foreach \src/\dest/\label in {{0/i/$a=b$},{i/ip/},{ip/N/}}
        \path[draw, thick, -stealth] (\src) -- node[above,sloped] {\label} (\dest);
      \foreach \src/\dest/\label in {{0/j/$a = b$},{j/jp/\wtab{$a = b = v$}},{jp/N/}}
        \path[draw, thick, -stealth] (\src) -- node[below,sloped] {\label} (\dest);
      \foreach \src/\dest in {{jp/i}}
        \path[draw, dashed] (\src) -- (\dest);
      \foreach \src/\dest in {{j/i}, {j/ip}, {jp/ip}}
        \path[draw, thick, -stealth] (\src) -- (\dest);
    \end{tikzpicture} \\
    (c) &
    \begin{tikzpicture}[xscale=1.2,baseline=(0.base)]
      \tikzstyle{vertex}=[draw,circle,minimum size=5pt,inner sep=0pt]
      \foreach \pos/\id/\lbl in {{(0, 0)/0/0}, {(3.1, 0)/N/N}}
      {
        \node[vertex] (\id) at \pos {};
        \path (\id) ++(0, -0.3) node {$\lbl$};
      }
      \node[vertex] (i) at (1, 0.5) {};
      \path (i) ++(-0.2, -0.3) node {$i$};
      \node[vertex] (ip) at (2.3, 0.5) {};
      \path (ip) ++(0.2, -0.3) node {$i{^+}$};
      \node[vertex] (j) at (1, -0.5) {};
      \path (j) ++(-0.2, -0.3) node {$i_t$};
      \node[vertex] (jp) at (2.3, -0.5) {};
      \path (jp) ++(0.2, -0.3) node {$i_t{^+}$};

      \foreach \src/\dest/\label in {{0/i/},{i/ip/},{ip/N/}}
        \path[draw, thick, stealth-stealth] (\src) -- node[above,sloped] {\label} (\dest);
      \foreach \src/\dest/\label in {{0/j/$a = b$},{j/jp/\wtab{$a = b = v$}},{jp/N/}}
        \path[draw, thick, -stealth] (\src) -- node[below,sloped] {\label} (\dest);
      \foreach \src/\dest in {{i/j}, {i/jp}, {j/ip}, {ip/jp}}
        \path[draw, thick, stealth-stealth] (\src) -- (\dest);
    \end{tikzpicture} 
    & \qquad &
    (f) &
    \begin{tikzpicture}[xscale=1.2,baseline=(0.base)]
      \tikzstyle{vertex}=[draw,circle,minimum size=5pt,inner sep=0pt]
      \foreach \pos/\id/\lbl in {{(0, 0)/0/0}, {(1, 0)/i/i}, {(2.3, 0)/ip/i{^+}}, {(3.1, 0)/N/N}}
      {
        \node[vertex] (\id) at \pos {};
        \path (\id) ++(0, -0.3) node {$\lbl$};
      }
      \foreach \src/\dest/\label in {{0/i/$a = b$},{i/ip/\wtab{}},{ip/N/}}
        \path[draw, thick, -stealth] (\src) -- node[above] {\label} (\dest);
    \end{tikzpicture}
    \\

  \end{tabular}
  \end{center}
  \caption{\label{fig-copyassign}
    Array content graph during analysis of the program given in Figure~\ref{fig-copy}, while
    executing $i = i+1$. 
The dashed edges indicate
    equality (in reality they represent two edges each labelled $\bot$).
  }
\end{figure}

\begin{example}
  Consider the array-copy program given in Figure~\ref{fig-copy}, immediately before the assignment
  $\mathtt{i = i+1}$.
  The state is shown in Figure~\ref{fig-copyassign}~(a). 

  We must first introduce a new variable $i_t$ to hold the prior value
  of $i$,
  transforming $i = i + 1$ to $i_t = i;\; i = i_t + 1$. The normalized graph
  after the $i_t = i$ statement is shown in Figure~\ref{fig-copyassign}~(b).
  To handle the $i = i_t + 1$ statement we must
  eliminate the annotations on edges corresponding
  to $i$ and $i{^+}$, resulting in state (c). Note that the edge $\psi_{0i_t{^+}}$, omitted from the diagram,
  has annotation $a = b$. State (d) completes the handling of the
  $i = i_t + 1$ statement, introducing the new value of $i$ into the
  scalar domain. The scalar domain discovers that $i_t{^+} = i$ (the dashed edge). When normalizing state (d),
  $\psi_{i_t{^+}i} = \bot$, so the rule $\psi_{0i} \bleq \psi_{0i_t{^+}} \join \psi_{i_t{^+}i}$ results in
  $\psi_{0i}$ becoming $a=b$, the desired invariant (e),
which after projecting out $i_t$ and $i_t^+$ gives (f).
\hfill \ebox
\end{example}

\subsubsection{Array reads.}
An array read $x = A[i]$ is relatively simple. As for standard variable
assignment, we must existentially quantify the variable $x$. But
instead of introducing a relation into the scalar domain, we add the constraint
$x = a$ to the singleton segment $\psi_{ii{^+}}$. This transfer function
may be formulated as
\longmath{
  $$\langle \varphi, \Psi \rangle \sem{\mathtt{x = A[i]}} = \langle \exists x.\varphi, \Psi' \rangle$$
\noindent}
\shortmath{$\langle \varphi, \Psi \rangle \sem{\mathtt{x = A[i]}} = \langle \exists x.\varphi, \Psi' \rangle$}
where $\Psi'$ is given by, for all $p,q$:
\begin{equation*}
  \psi'_{pq} = \left\{
    \begin{array}{cl}
      \psi_{pq} \sem{\mathtt{a = x}}~~ & \mathbf{if}~ p = i, q = i{^+} \\
      \exists x~.~\psi_{pq} & \mathbf{otherwise}
    \end{array}
    \right.
\end{equation*}
Normalization handles the consequences for the scalar part.

\subsubsection{Array writes.}
When we store a value into an array at index $i$, we update the 
corresponding edge property $\psi_{ii{^+}}$. 
However, this is not sufficient, as the singleton
$i$ may be covered by other edges. 
As with previous analyses, we distinguish between
\emph{strong updates} where all elements in a (generally, singleton) 
segment are updated with a given property, and 
\emph{weak updates}~\cite{Blanchet_Smashing} where some elements
of a segment \emph{may} be updated.
An edge $\psi_{pq}$ must be updated if $p \leq i \leq q \wedge p < q$ 
is possible in the current state. 
This is possible if and only if the edges $\psi_{pi{^+}}$ and 
$\psi_{iq}$ are both feasible.

Consider the array state illustrated in Figure~\ref{fig-write}, where
$0 \leq i \leq j \leq N$ and the first $j$ elements have been initialized to 0.
When we store $1$ at $A[i]$, we update the singleton $\psi_{ii{^+}}$ with the property
$a = 1$. 
However, there are other segments that may contain $A[i]$. The segments $\psi_{0i{^+}}$,
$\psi_{ij}$ and $\psi_{0j}$ are all consistent with index $i$. In this case, they must
all be weakly updated (annotations on $\psi_{0i{^+}}$ and $\psi_{ij}$ have been omitted; they
are identical to the annotation on $\psi_{0j}$).

\begin{figure}[t]
  \begin{center}
    \begin{tabular}{ccc}
    $S$ &
    \begin{tikzpicture}[xscale=1.2, baseline=(0.base)]
      \tikzstyle{vertex}=[draw,circle,minimum size=5pt,inner sep=0pt]
      \foreach \pos/\id/\lbl in {{(0, 0)/0/0}, {(1, 0.5)/i/i}, {(2, 0.5)/ip/i{^+}}, {(3, 0)/j/j}, {(3.5, 0)/N/N}}
      {
        \node[vertex] (\id) at \pos {};
        \path (\id) ++(0, -0.3) node {$\lbl$};
      }
      \foreach \src/\dest/\label in {{0/i/$a = 0$},{i/ip/$a = 0$},{ip/j/$a = 0$},
                                     {j/N/}}
        \path[draw, thick, -stealth] (\src) -- node[above,sloped] {\label} (\dest);
      \path[draw, thick, -stealth,bend right] (0) to node[below] {$a = 0$} (ip);
      \path[draw, thick, -stealth,bend right] (i) to node[below] {$a = 0$} (j);
      \path[draw, thick, -stealth] (0) .. controls (1, -1) and (2, -1) .. node[below] {$a = 0$} (j);
    \end{tikzpicture}
    &
    \begin{tabular}{c|ccccc}
     & $0$ & $i$ & $i^{+}$ & $j$ & $N$ \\ 
    \hline
    $0$ & $\bot$ & $a = 0$ & $a = 0$ & $a = 0$ & $\top$ \\
    $i$ & $\bot$ & $\bot$ & $a = 0$ & $a = 0$ & $\top$ \\
    $i^{+}$ & $\bot$ & $\bot$ & $\bot$ & $a = 0$ & $\top$ \\
    $j$ & $\bot$ & $\bot$ & $\bot$ & $\bot$ & $\top$ \\
    $N$ & $\bot$ & $\bot$ & $\bot$ & $\bot$ & $\bot$ \\
    \end{tabular}
    \\
    \\
    $S\sem{\mathtt{A[i] = 1}}$ &
    \begin{tikzpicture}[xscale=1.2, baseline=(0.base)]
      \tikzstyle{vertex}=[draw,circle,minimum size=5pt,inner sep=0pt]
      \foreach \pos/\id/\lbl in {{(0, 0)/0/0}, {(1, 0.5)/i/i}, {(2, 0.5)/ip/i{^+}}, {(3, 0)/j/j}, {(3.5, 0)/N/N}}
      {
        \node[vertex] (\id) at \pos {};
        \path (\id) ++(0, -0.3) node {$\lbl$};
      }
      \foreach \src/\dest/\label in {{0/i/$a = 0$},{i/ip/$a = 1$},{ip/j/$a = 0$},
                                     {j/N/}}
        \path[draw, thick, -stealth] (\src) -- node[above,sloped] {\label} (\dest);
      \path[draw, thick, -stealth,bend right] (0) to node[below=3] {\ldots} (ip);
      \path[draw, thick, -stealth,bend right] (i) to node[below=3] {\ldots} (j);
      \path[draw, thick, -stealth] (0) .. controls (1, -1) and (2, -1) .. node[below] {$a = 0 \join a = 1$} (j);
    \end{tikzpicture}
    &
    \begin{tabular}{c|ccccc}
     & $0$ & $i$ & $i^{+}$ & $j$ & $N$ \\ 
    \hline
    $0$ & $\bot$ & $a = 0$
        & $\left(\begin{array}{c} a = 0 \\ \join \\ a = 1\end{array}\right)$
        & $\left(\begin{array}{c} a = 0 \\ \join \\ a = 1\end{array}\right)$ & $\top$ \\
    $i$ & $\bot$ & $\bot$ & $a = 1$ & $\left(\begin{array}{c} a = 0 \\ \join \\ a = 1\end{array}\right)$ & $\top$ \\
    $i^{+}$ & $\bot$ & $\bot$ & $\bot$ & $a = 0$ & $\top$ \\
    $j$ & $\bot$ & $\bot$ & $\bot$ & $\bot$ & $\top$ \\
    $N$ & $\bot$ & $\bot$ & $\bot$ & $\bot$ & $\bot$ \\
    \end{tabular}

    \end{tabular}
  \end{center}
  \caption{\label{fig-write}
    In state $S$, we know that all elements between $0$ and $j$ have been initialized to $0$.
    When we evaluate $\mathtt{A[i] = 1}$, we update the edge $\psi_{ii{^+}}$. However, $i$ may
    also be covered by the edge $\psi_{0j}$. For these potentially overlapping edges, we must
    perform a \emph{weak update}, taking the join of the previous value with the new.
  }
\end{figure}
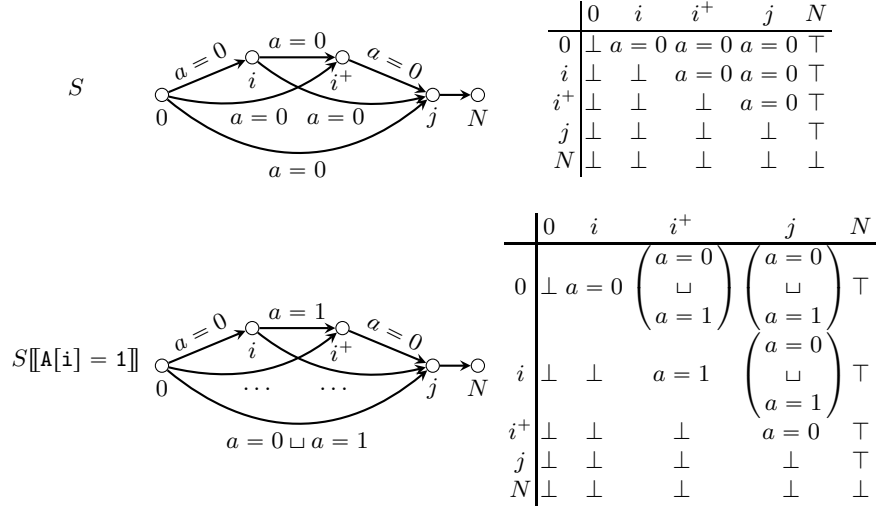
We can formulate this as $\langle \varphi, \Psi \rangle \sem{\mathtt{A[i] = f}}
= \langle \varphi, \Psi' \rangle$, where $\Psi'$ is given by, for all $p,q$:
\begin{equation*}
  \psi'_{pq} = \left\{
    \begin{array}{cl}
      \psi_{pq} \sem{\mathtt{a = f}} & \mathbf{if}~ p = i, q = i{^+} \\
      \psi_{pq} & \mathbf{if}~\varphi \models (p > i \vee q < i{^+}) \\
      \psi_{pq} \join \psi_{pq} \sem{\mathtt{a = f}}~~ & \mathbf{otherwise}  \\
    \end{array}
    \right.
\end{equation*}

\noindent
Notice that if we have some other variable $j$ such that $\varphi \models i = j$, we will initially perform
only a weak update of the segment $\psi_{jj{^+}}$.
However, the normalization procedure will
then enforce consistency between $\psi_{ii{^+}}$ and $\psi_{jj{^+}}$.

\section{Improving Efficiency through Relaxation}
\label{sec-efficiency}
When computing the strongest matrix entries, we perform a substantial amount of redundant work. We
include all constraints $\varphi$ from the scalar domain in each matrix entry $\psi_{ij}$, and these constraints
will be
processed during each step in the shortest-path computation. 
This is not ideal, as many are irrelevant to the content properties, 
and abstract domain operations are often proportional to the 
number of constraints. 
Hence we want
to construct some relaxation $\ddot{\psi}_{ij}$ of $\psi_{ij}$ that 
discards irrelevant
scalar properties. 
We shall use $\odot$ to denote this relaxation operation (that is,
$\ddot{\psi}_{ij} = \varphi \odot \psi_{ij}$). 
At a minimum, it must satisfy:
  $$\psi_{ij} \bleq \varphi \meet \ddot{\psi}_{ij} = \varphi \meet (\varphi \odot \psi_{ij})$$
We want to make $\varphi \meet \ddot{\psi}_{ij}$ as close to $\psi_{ij}$ 
as possible, 
while keeping the representation of $\ddot{\psi}_{ij}$ concise.
This is very similar to the process of constraint 
abduction~\cite{Maher_LICS05}.
However, even with a relaxation $\ddot{\Psi}$ such that $\varphi \meet \ddot{\psi}_{ij} \equiv \psi_{ij}$,
we may still lose relevant information.

\begin{example}
\label{ex-abloss}
Consider a program state (using the domain of \texttt{octagons}) 
with scalar property $\varphi = \sem{x < y}$,
and segment properties $\psi_{ik} = \varphi \meet \sem{i < k \wedge x < a}$, $\psi_{kj} = \varphi \meet \sem{k < j \wedge y = a}$.
Computing the value of $\psi_{ij}$ gives the expected
$\psi_{ij} = \varphi \meet \sem{i < j \wedge x < a}$.
However, although the relaxations 
\[
  \ddot{\psi}_{ik} = \sem{i < k \wedge x < a}
  \quad\mathrm{and}\quad
  \ddot{\psi}_{kj} = \sem{k < j \wedge y = a}
\]
are exact, computing $\psi_{ij}$ as before yields:
\[
  \sem{i < j \wedge i < k \wedge x < a} \join \sem{i < j \wedge k < j
    \wedge y = a} = \sem{i < j} \hfill \ebox
\]
\end{example}
We could avoid this loss of information by conjoining the 
scalar part during each step of the fixed point computation
  $$\ddot{\psi_{ij}} \bleq \varphi \odot
    ((\varphi \meet \ddot{\psi}_{ij} \meet \ddot{\psi}_{ik}) \join (\varphi \meet \ddot{\psi}_{ij} \meet \ddot{\psi}_{kj}))$$
but while this avoids the loss of information, it also defeats the 
original goal of reducing computation cost.
Instead we define a more conservative $\odot$ operation which maintains enough additional information
to retain properties of interest.

\begin{example}\label{ex-abkeep}
Consider again the analysis that was performed in Example~\ref{ex-abloss},
but now with a modified $\odot$ operation which gives us 
$\ddot{\psi}_{ik} = \sem{i < k \wedge x < a}$ and 
$\ddot{\psi}_{kj} = \sem{k < j \wedge y = a \wedge x < a}$.
In this case, when we compute $\ddot{\psi}_{ij}$, we get:
\longmath{
  $$\ddot{\psi_{ij}} = \sem{i < j \wedge i < k \wedge x < a} \join \sem{i < j \wedge k < j \wedge y = a \wedge x < a}
    = \sem{i < j \wedge x < a}$$}
\shortmath{
  $\ddot{\psi_{ij}} = \sem{i < j \wedge i < k \wedge x < a} \join \sem{i < j \wedge k < j \wedge y = a \wedge x < a}
    = \sem{i < j \wedge x < a}$.}
This maintains the property of interest, without repeatedly conjoining 
with $\varphi$ during the fixed point process.
\hfill \ebox
\end{example}

Another observation from Examples~\ref{ex-abloss} and~\ref{ex-abkeep} 
is that while the boundary constraints for each edge 
(such as $\sem{i < j}$ for $\psi_{ij}$) 
are not implied by $\varphi$, they are typically irrelevant
to the segment properties, unless:
\begin{itemize}
  \item $\varphi \meet \sem{i < j} = \bot$, in which case the edge must be empty, or
  \item the array content $A[\ell]$ is some function of the index $\ell$
  \item $\psi_{ji^{+}} = \bot$ (or $j = k^{+}$ for some $k$, and $\psi_{ki} = \bot$)
\end{itemize}
This is particularly troublesome for domains that explicitly store the transitive closure of constraints,
as we then expend substantial computation maintaining consequences of $\sem{i < j}$, which are largely irrelevant, and lost during the join.

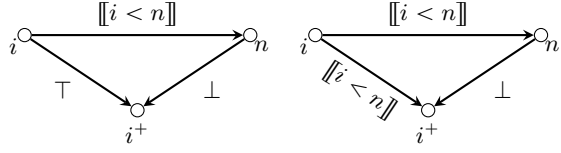
\begin{figure}[t]
\begin{center}
  \begin{tabular}{ccc}
  \begin{tikzpicture}[baseline=(0.base)]
    \tikzstyle{vertex}=[draw,circle,minimum size=5pt,inner sep=0pt]
    \foreach \pos/\off/\id/\lbl in
         {{(0, 0)/(-0.15, -0.15)/i/i},
          {(1.5, -1)/(0, -0.3)/ip/i{^+}},
          {(3, 0)/(0.15, -0.15)/n/n}}
    {
      \node[vertex] (\id) at \pos {};
      \path (\id) ++\off node {$\lbl$};
    }
    \foreach \src/\dest/\label in {{i/n/$\sem{i < n}$}}
      \path[draw, thick, -stealth] (\src) -- node[above] {\label} (\dest);
    \foreach \src/\dest/\label in {{n/ip/$\bot$}}
      \path[draw, thick, -stealth] (\src) -- node[below right] {\label} (\dest);
    \foreach \src/\dest/\label in {{i/ip/$\top$}}
      \path[draw, thick, -stealth] (\src) -- node[below left] {\label} (\dest);
  \end{tikzpicture}
  & \qquad &
  \begin{tikzpicture}[baseline=(0.base)]
    \tikzstyle{vertex}=[draw,circle,minimum size=5pt,inner sep=0pt]
    \foreach \pos/\off/\id/\lbl in
         {{(0, 0)/(-0.15, -0.15)/i/i},
          {(1.5, -1)/(0, -0.3)/ip/i{^+}},
          {(3, 0)/(0.15, -0.15)/n/n}}
    {
      \node[vertex] (\id) at \pos {};
      \path (\id) ++\off node {$\lbl$};
    }
    \foreach \src/\dest/\label in {{i/n/$\sem{i < n}$}}
      \path[draw, thick, -stealth] (\src) -- node[above] {\label} (\dest);
    \foreach \src/\dest/\label in {{n/ip/$\bot$}}
      \path[draw, thick, -stealth] (\src) -- node[below right] {\label} (\dest);
    \foreach \src/\dest/\label in {{i/ip/$\sem{i < n}$}}
      \path[draw, thick, -stealth] (\src) -- node[below,sloped] {\label} (\dest);
  \end{tikzpicture}
  \end{tabular}
\end{center}
\caption{\label{fig-loss}
  Solving this set of constraints, we strengthen $\psi_{ii{^+}}$ and
  derive $\sem{i < n}$. If we use $\odot$,
  we have $\ddot{\psi}_{in} = \top$, so we fail to strengthen $\ddot{\psi}_{ii{^+}}$.}
\end{figure}

If we choose an operator $\odot$ which discards the consequences of 
edge boundaries, we must take particular care not to lose information in
the third case mentioned.
Consider the state shown in Figure~\ref{fig-loss},
where the edge from $n$ to $i{^+}$ is infeasible. Computing
the original fixed point, we obtain $\sem{i < n}$ at $\psi_{ii{^+}}$, 
which is then lifted out to the scalar domain. 
However, $\sem{i < n}$ is obviously implied by 
$\varphi \meet \sem{i < n}$, so will typically be discarded by $\odot$. 
The property $\sem{i < n}$ is not obtained at $\psi_{ii^{+}}$, 
so is never lifted out to the scalar domain.
We must, therefore, add the following case to the normalization rules 
given in Section~\ref{sec-normal}:
\begin{enumerate}
\setcounter{enumi}{3}
\item
  $\varphi \bleq \sem{i \geq j} ~\textbf{if}~ \psi_{i, j} = \bot$
\end{enumerate}
This ensures that any scalar properties resulting from infeasible 
segments are included in the scalar domain.
Further, since many elements of $\ddot{\Psi}$ may be $\top$,
the relaxation operator allows us to take advantage of
sparse matrix representations.

Unfortunately we are not aware of any existing, general operations 
suitable for computing $\ddot{\Psi}$;
as they must consider both
underlying lattice and the characteristics of the implementation.
We complete this section by outlining suitable relaxation operators 
for difference-bound matrices (DBMs)~\cite{dill-dbm} 
(and \texttt{octagons}).
Relaxations for \texttt{polyhedra} can be found in the appendix, 
and it should not be particularly difficult to 
define analogous operators for alternative domains.

A value in the \texttt{DBM} (or \texttt{octagon}) domain consists of a set of constraints
$\sem{v_i - v_j \leq k}$ (or $\sem{\pm v_i \pm v_j \leq k}$ for \texttt{octagon}).
We can construct a relaxation of $\psi_{ij}$ by computing the transitive closure $\psi^{\star}_{ij}$,
and discarding any constraints implied by $\varphi \meet \sem{i < j}$:
  $$\varphi \odot \psi_{ij} = \{ c ~|~ c \in \psi^{\star}_{ij},~\varphi \meet \sem{i < j} \not \vdash c \}$$
If the abstract states are stored in closed form,
we can simply collect the constraints of $\psi_{ij}$ not appearing in $\varphi \meet \sem{i < j}$.
Or we can avoid performing many implication tests by instead collecting the constraints
involving variables in $U$.

\section{Experimental Evaluation}
\label{sec-exper}
We have implemented the analysis in \sys{sparcolyzer}, a prototype array content analyser for the language described
in Section~\ref{sec-lang}. \sys{sparcolyzer} is implemented in \sys{ocaml}, using the \sys{Fixpoint}
library\footnote{\url{http://pop-art.inrialpes.fr/people/bjeannet/bjeannet-forge/fixpoint/}}.
For the underlying domain, we implemented the DBM domain in C++,
customized for operating on sparse graphs.
Experiments were performed on a 3.0GHz Core 2 Duo with 4Gb ram running Ubuntu Linux 12.04.
The set \Va\ of segment bounds were pre-computed using a simple data-flow analysis to collect
all variables which may (possibly indirectly) be involved in the computation of an array index.

We tested \sys{sparcolyzer} on a number of array manipulation program fragments.
Most of these were taken from Halbwachs and P{\'e}ron 
\cite{Halbwachs_Peron_PLDI08}. 
We added several additional fragments that illustrate 
interesting properties of the analysis, including
members of the \texttt{init\_rand} family discussed in 
Section~\ref{sec-intro} (Figure~\ref{fig-initrand}).

Computation time for each instance is given in Table~\ref{tab:times}.
For instances taken from \cite{Halbwachs_Peron_PLDI08}, 
we include the original reported runtimes,
although these are not directly comparable, as the experiments
in \cite{Halbwachs_Peron_PLDI08} were performed on a
slower machine (Core2 Duo 1.6 GHz, with 2MB of RAM)
using the domain of 
\emph{difference bound matrices with disequalities
(\texttt{dDBM})}~\cite{Halbwachs_Peron_VMCAI07}.%
\footnote{Performance of the domains should be roughly equivalent, 
as in the absence of explicit disequalities, \texttt{dDBM} behaves 
identically to \texttt{DBM}.}
The implementation from \cite{Halbwachs_Peron_PLDI08} has not been
available and we have not tried to reconstruct it.

Table~\ref{tab:times} compares the runtimes of two variants of our
content domain to the approach of Halbwachs and P{\'e}ron
\cite{Halbwachs_Peron_PLDI08} (the hp$_{08}$ column).
The \texttt{naive} variant
uses the direct implementation, where the matrix is represented as a 
$\abs{V'}\times\abs{V'}$ array, 
and a copy of the scalar domain is stored in each matrix entry.
The \texttt{sparse} variant stores, for each row and column, 
a set of non-$\top$ entries so that the normalization, $\join$ and 
$\meet$ operations do not need to process entries that will definitely
remain $\top$, and computes the fixed point on the relaxed matrix
$\ddot{\Psi}$, rather than directly on $\Psi$. 
\texttt{sparse} uses the simple relaxation step of discarding all 
constraints not involving some array variable $A \in U$.
Note that \texttt{sparse} still iterates over all 
$\abs{V'}\times\abs{V'}$ elements when changes to the scalar domain
occur, as a change to the scalar domain may affect any matrix element.

In cases where there are very few partitions---either because there are very few index
variables, or they are highly constrained---we expect the partition-based methods to be faster
(as they do not need to compute closure over transitive edges).
The performance of \texttt{naive} is comparable
to that of \cite{Halbwachs_Peron_PLDI08} on instances with few 
partitions, and it improves substantially on more complex instances.
\texttt{sparse} is faster yet, sometimes by several orders of magnitude, and still
finds the desired invariant in all but two
cases.

\begin{table}[t]
\begin{center}
\setlength{\tabcolsep}{9pt}
\begin{tabular}{||l||c|c||c||}
  \hhline{|t:=:t:==:t:=:t|}
  \textbf{program} & \texttt{naive} & \texttt{sparse} & \texttt{hp$_{08}$ \cite{Halbwachs_Peron_PLDI08}} \\
  \hhline{||-||-|-||-||}
  \tt init & 0.11 & 0.02 & \\
  \tt init\_offset & 0.28 & 0.07 & \hphantom{2}0.05 \\
  \tt init\_rand$_2$ & 2.48 & 0.14 & \\
  \tt init\_rand$_3$ & 9.59 & 0.64 & \\
  \tt init\_rand$_4$ & 31.14 & 1.98 & \\
  \tt init\_rand$_5$ & 80.58 & 4.96 & \\
  \tt arraymax & 0.13 & \llap{$<$}0.01  & \hphantom{2}0.10 \\
  \tt copy & 0.13 & \llap{$<$}0.01 & \hphantom{2}0.02 \\
  \tt partition\_hoare & 1.51 & 0.06 & \\
  \tt partition\_hp08 & 3.50 & 0.14\rlap{$^\dagger$} & 22.87 \\
  \tt sentinel & 0.14 & \llap{$<$}0.01 & \hphantom{2}0.21 \\
  \tt first\_nonnull & 0.60\rlap{$^\dagger$} & 0.01\rlap{$^\dagger$} & \hphantom{2}2.25 \\
  \hhline{|b:=:b:==:b:=:b|}
\end{tabular}
\end{center}
\caption{\label{tab:times}
  Analysis times in seconds.
  Instances where we were unable to prove the desired
  invariant are marked with $\dagger$.}
  \vspace*{-2em}
\end{table}

It is interesting to compare the behaviour of \texttt{sentinel} and \texttt{first\_nonnull}. These programs superficially
appear quite similar; in both cases, we set up an `end-of-array' marker, then scan the array to find a particular element.
However, the invariants necessary to prove the desired properties are quite different.
In the case of \texttt{first\_nonnull}, we require:
  $$(s = n \wedge \forall e \in [0, n) ~.~ A[e] \neq 0)
     \vee (s < n \wedge A[s] = 0 \wedge \forall e \in [0, s) ~.~ A[e] \neq 0)$$
This can be expressed using the approaches of
Gopan \emph{et al.} and Halbwachs
and P{\'e}ron~\cite{Gopan_POPL05,Halbwachs_Peron_PLDI08},
as they store a separate invariant for each
total ordering amongst the partition variables. 
Our approach, however, cannot handle such disjunctive reasoning,
so the segment property quickly reaches $\top$.
\begin{figure}[t]
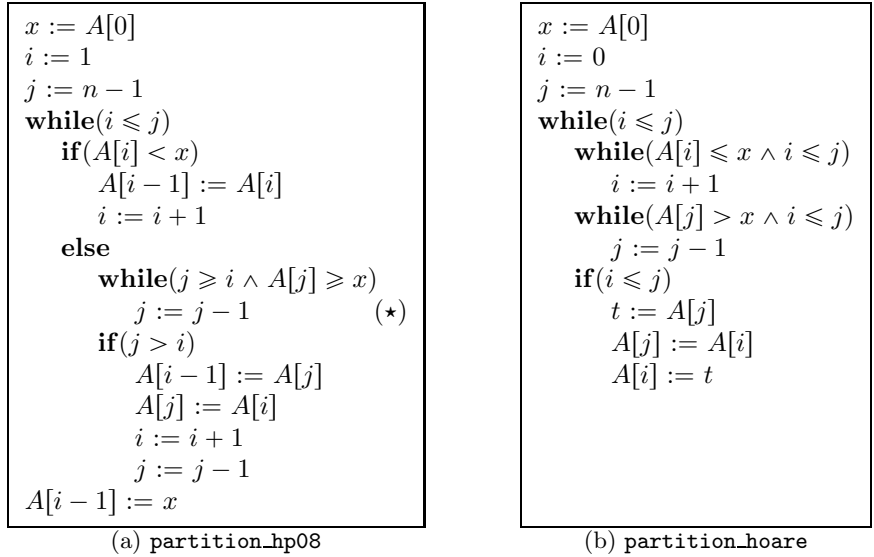

\centerline{
  \setlength{\tabcolsep}{15pt}
  \begin{tabular}{cc}
  \pcode{
    $x$ := $A[0]$ \\
    $i$ := $1$ \\
    $j$ := $n-1$ \\
    \textbf{while}($i \leq j$) \\
    \> \textbf{if}($A[i] < x$) \\
    \> \> $A[i-1]$ := $A[i]$ \\
    \> \> $i$ := $i+1$ \\
    \> \textbf{else} \\
    \> \> \textbf{while}($j \geq i \wedge A[j] \geq x$) \\
    \> \> \> $j$ := $j-1$ ~~~~~~~~~~~~~($\star$) \\
    \> \> \textbf{if}($j>i$) \\
    \> \> \> $A[i-1]$ := $A[j]$ \\
    \> \> \> $A[j]$ := $A[i]$ \\
    \> \> \> $i$ := $i+1$ \\
    \> \> \> $j$ := $j-1$ \\
    $A[i-1]$ := $x$
  }
  &
  \pcode{
    $x$ := $A[0]$ \\
    $i$ := $0$ \\
    $j$ := $n-1$ \\
    \textbf{while}($i \leq j$) \\
    \> \textbf{while}($A[i] \leq x \wedge i \leq j$) \\
    \> \> $i$ := $i+1$ \\
    \> \textbf{while}($A[j] > x \wedge i \leq j$) \\
    \> \> $j$ := $j-1$ \\
    \> \textbf{if}($i \leq j$) \\
    \> \> $t$ := $A[j]$ \\
    \> \> $A[j]$ := $A[i]$ \\
    \> \> $A[i]$ := $t$ \\
    \\
    \\
    \\
  } \\
  (a) \texttt{partition\_hp08} & (b) \texttt{partition\_hoare}
  \end{tabular}
}
\caption{\label{fig-part}
  Quicksort partitioning:
  (a) as done in~\cite{Halbwachs_Peron_PLDI08};  and (b) a la Hoare
}
\end{figure}

Consider \texttt{partition\_hp08}
(the variant of the Quicksort partition step given
in~\cite{Halbwachs_Peron_PLDI08}, shown in Figure~\ref{fig-part}(a)).
As the imperative source
language we use does allow loads inside conditionals, such reads
are hoisted outside the corresponding loops; for example, the
loop marked $(\star)$ is transformed as shown below.
\begin{wrapfigure}[8]{r}{0.35\textwidth}
\vspace*{-1ex}
\centerline{
\pcode{
  $ej$ := $A[j]$ \\
  \textbf{while}($j \geq i \wedge ej \geq x$) \\
  \> $j$ := $j-1$ \\
  \> $ej$ := $A[j]$
}
}
\end{wrapfigure}
At the point marked $(\star)$, it is possible to determine that 
$\sem{j < i} \Rightarrow \sem{A[j] < x}$.  Thus
it is easy to show that $\sem{A[j] < x}$ holds at the loop exit. 
In the hoisted version, \texttt{naive} method can prove the invariant 
successfully, because the property $\sem{j < i} \Rightarrow \sem{ej < x}$
is derived for the edge $\psi_{ji}$. When we exit the loop with 
$\sem{j < i}$, this property gets extracted to the scalar domain, 
and we get $\sem{ej < x} \wedge \sem{ej = a}$ at $\psi_{jj'}$.
When using the \texttt{sparse} method, however, 
the property on $\psi_{ji}$ is discarded, 
as it involves only scalar variables, so the invariant is lost.

If we were to use the original version without hoisting, we would be
unable to prove the invariant using either method, as we cannot
express $\sem{j < i} \Rightarrow \sem{A[j] < x}$ directly.

However, our method easily proves the standard version
(\texttt{partition\_hoare}; Figure~\ref{fig-part}(b)) correct,
whether or not the reads are hoisted.

\section{Conclusion and Future Work}
\label{sec-future}
We have described a new approach to automatic discovery of array 
properties, inspired by algebraic shortest-path algorithms.
This approach retains much of the expressiveness of the partitioning
methods of~\cite{Gopan_POPL05} and~\cite{Halbwachs_Peron_PLDI08},
but avoids the need for syntax dependence and an up-front factorial
partitioning step.
The method can successfully derive invariants for a range of interesting
array program fragments,
and is substantially faster than partitioning-based approaches
with even modest numbers of index variables.

Several improvements could be made to the performance of the analysis.
The current implementation does not take advantage of liveness 
information, and maintains entries in the content graph for all 
variables in $\Va$ at each step. Clearly, performance could be
improved by eliminating non-live variables from the matrix.

Algorithms which maintain shortest path information can often be 
improved by storing only the \emph{transitive reduction} of the graph.
As our domains are not distributive, it is non-trivial to
determine whether a given edge must occur in the transitive reduction; 
however, it would be worth investigating whether maintaining the 
transitive reduction would prove beneficial.

\bibliographystyle{abbrv}
\bibliography{refs}

\newpage
\appendix

\section*{Appendix A: Termination Proofs}\label{appendix-term}
\begin{theorem}
The shortest-path computation terminates for the octagon domain.
\end{theorem}
\begin{proof}
Consider an initial set of abstract states $X = [x_1, \ldots, x_n]$, 
and a system of inequalities of the form 
$x_i \bleq (x_i \meet x_j) \join (x_i \meet x_k)$.
Let $\mathcal{BP}(x)$ denote the bounding hyperplanes of $x$. 
We then have:
    $\mathcal{BP}((x \meet y) \join (x \meet z)) \subseteq
        \mathcal{BP}(x \meet y)
        \cup \mathcal{BP}(x \meet z)$.
Each time one of the equations is evaluated, each bounding hyperplane 
is an element of:
\longmath{
    $$\bigcup_{X' \subseteq X} \mathcal{BP}(\glb X')$$
}
\shortmath{$\bigcup_{X' \subseteq X} \mathcal{BP}(\glb X')$}
As the initial set of bounding hyperplanes is finite, and each 
iteration must tighten at least one bounding plane, 
the tightening process must eventually terminate; 
in fact, the number of descending steps is bounded by:
\longmath{
\[
    \abs{X} \abs{\bigcup_{x \in X} \mathcal{BP}(x)}
\quad \qed
\]
}
\shortmath{$\abs{X} \times \abs{\bigcup_{x \in X} \mathcal{BP}(x)}$
\qed}
\end{proof}
In the case of octagons, projection cannot introduce new bounding 
hyperplanes, so the addition of the propagation rules 
($\varphi \bleq \exists~U_I~.~\varphi_{ii^{+}}$ and $\varphi_{ij} \bleq \psi$)
does not affect termination.

\begin{theorem}
\label{thm-poly}
The shortest-path computation terminates for the convex polyhedron domain.
\end{theorem}
\begin{proof}
We can prove termination for convex polyhedra in a similar fashion 
as for octagons.
Consider an initial set of abstract states $X = [x_1, \ldots, x_n]$, 
and a system of inequalities of the form 
$x_i \bleq (x_i \meet x_j) \join (x_i \meet x_k)$.
By the polyhedron decomposition theorem (see, e.g.~\cite{Schrijver}), 
  any polyhedron may be generated
  by a finite set of points $P$, and a set of rays $R$.
  Let $R_V$ denote the set of rays of unit length in the direction of variable $v$ (that is,
  a vector with $v^{th}$ component $1$ or $-1$, and all other components $0$), for each $v \in V$.
  Given polyhedra $x = \langle P_x, R_x \rangle$ and $y = \langle P_y, R_y \rangle$, the results
  of the operations of interest have the following properties:
    $$
      \begin{array}{rcl}
        x \meet y = \langle P', R' \rangle & ~\textbf{where}~& R' \subseteq R_x \cup R_y \\
        x \join y = \langle P', R' \rangle & ~\textbf{where}~& P' \subseteq P_x \cup P_y,~R' \subseteq R_x \cup R_y \\
        \exists~V~.~x = \langle P', R' \rangle & ~\textbf{where}~ & P' \subseteq P_x,~R' \subseteq R_x \cup R_V
      \end{array}$$

\noindent
None of these operations can introduce rays not in $\mathcal{R} = R_V \cup \bigcup_{x \in X} R_x$. Additional
  extreme points can only be introduced during the application of $\meet$.
  During each fixed point iteration, exactly one of three cases must occur:
  \begin{enumerate}
    \item All abstract values in $X$ remain the same.
    \item The set of rays for some $x \in X$ changes.
    \item All the sets of rays remain the same, and the set of extreme points for some $x \in X$ changes.
  \end{enumerate}
  In case $1$, we terminate. As the set of rays is restricted to $\mathcal{R}$, and each iteration is
  strictly descending, case $2$ can only occur finitely many times.
  Assuming the set of rays remains fixed, every introduced extreme point
  must be some element of
\longmath{
    $$\bigcup_{X' \subseteq X} \{ P ~|~ \langle P, R \rangle = \bigcap_{x
      \in X'} x \}$$
}
\shortmath{
$\bigcup_{X' \subseteq X} \{ P ~|~ \langle P, R \rangle = \bigcap_{x
      \in X'} x \}$.
}
  As this set is finite, and each step is descending, this can only occur finitely many times without
  case $2$ occurring.
  As case $2$ must always occur after a bounded number of steps, and can only occur finitely many times,
  the fixed point process must eventually terminate.
\qed
\end{proof}


We conclude that this process will terminate on the most commonly 
used relational numeric domains.
In the case of polyhedra, it is worth noting that Theorem~\ref{thm-poly}
does not provide any bounds on the \emph{coefficients} of 
hyperplanes in the resulting polyhedra. 
In some cases, the coefficients may grow quite large before converging,
which can cause problems for domains implemented with 
fixed-precision machine arithmetic.

\section*{Appendix B: Relaxation for \texttt{polyhedra}}
\label{appendix-polyhedra}

The relaxation algorithm for \texttt{polyhedra} follows the same 
intuition as that for \texttt{octagons};
we wish to collect the transitive closure of $\psi_{ij}$, 
then discard anything
implied separately by $\varphi \meet \sem{i < j}$.

However, the \texttt{polyhedra} domain provides two difficulties: 
computing the transitive closure
of a set of linear constraints is non-trivial (as the domain is typically stored as a minimal set of
generators and hyperplanes~\cite{Cousot_Halbwachs_POPL78});
and often is a bad idea in general: as the constraints do not
have bounded arity, there may be exponentially more constraints in 
the transitive closure than in the original problem.
Instead, we collect a set of constraints that is sufficient to 
reconstruct the constraints that, had we computed the transitive 
closure, would have been kept.

Given constraints $c_1 = \sem{k_1^\intercal x \geq m_1}$ and $c_2 = \sem{k_2^\intercal x \geq m2}$,
we can construct a new constraint (that is not implied separately by either $c_1$ or $c_2$) by
resolution if there is some matched pair of coefficients
$k_{1v}$, $k_{2v}$, such that $k_{1v} > 0$ and $k_{2v} < 0$.
We then construct a new constraint
$c_1 = \sem{(k_1 + \frac{k_{1v}}{k_{2v}}k_2)^\intercal x \geq m_1 + \frac{k_{1v}}{k_{2v}}m_2}$,
which has $v^{th}$ coefficient $0$.

\begin{example}\label{ex-resolve}
  Consider the constraints $c_1 = \sem{x + y \geq 7}$, $c_2 = \sem{z - 2y \geq 2}$, and
  $c_3 = \sem{w + 2y \geq 3}$.
  $c_1$ and $c_2$ may be resolved, as $c_1$ contains the term $y$, and $c_2$ contains $-2y$.
  This yields $c_{12} = \sem{x + \frac{z}{2} \geq 8}$.
  We cannot, however, construct any new constraints by combining $c_1$ with $c_3$.
\hfill \ebox
\end{example}

Rather than computing the transitive closure explicitly, given an initial set of \emph{interesting}
constraints $\ddot{C}$ and other constraints $C$, we find all those constraints that are \emph{resolvable}
with those in $\ddot{C}$ (taking into account the direction of previous
resolution steps), and add them to $\ddot{C}$. 
We then continue this process until no further resolvable constraints are found:
  $$\mathsf{trans^{\star}}(R, \ddot{C}, C) = \left\{ \begin{array}{l}
      \mathsf{trans^{\star}}(R \cup \{(\mathsf{sign}(k_{v'}), v')\}, \ddot{C} \cup \{c\}, C) \\
      \hphantom{\ddot{C}}~~\textbf{if}~\exists~(s,v) \in R,~c = \sem{k^\intercal x \geq m} \in C,~v' \in x \\
      \hphantom{\ddot{C}}~~\textbf{s.t.}~v \neq v',~k_v s < 0,
      ~((\mathsf{sign}(k_{v'}), v') \notin R \vee c \notin \ddot{C} \\
      \ddot{C} ~\textbf{otherwise}
    \end{array} \right. $$
  $$\mathsf{trans}(\ddot{C}, C) = 
  \mathsf{trans^{\star}}(\{ \mathsf{sign}(k_v, v) ~|~ \sem{k^\intercal x \geq m} \in \ddot{C},~v \in x \}, \ddot{C}, C)$$
Given a set of linear constraints $C$ defining a polyhedron, 
we can construct the initial set $\ddot{C}$ with the elements of $C$ that are not
implied by $\varphi \meet \sem{i < j}$, then compute the relaxation with
$\mathsf{trans}(\ddot{C}, C~\setminus~\ddot{C})$.
As with \texttt{octagons}, we can avoid performing implication tests by instead initializing
$\ddot{C}$ with those constraints containing variables in $V_{\mathcal{A}}$.

\begin{example}\label{ex-polyextract}
  Consider a constraint of interest $c_0 = \sem{a - y \geq 0}$, with additional constraints
  $C = \{c_1 = \sem{y - z \geq 0}, c_2 = \sem{z + w \geq 0}, c_3 = \sem{x - y \geq 0}\}$.
  Initially, we have $\ddot{C} = \{\sem{a - y \geq 0}\}$, and $R = \{(+, a), (-, y)\}$.
  We can resolve with $\sem{y - z \geq 0}$, since $(-, y) \in R$, but $(-, z) \notin R$.
  During the second step, we include $c_2$, adding $(-, w)$ to $R$.
  At this point, we have $R = \{(+, a), (-, y), (-, z), (+, w)\}$, and
  $\ddot{C} = \{c_0, c_1, c_2\}$, 
and we cannot add anything to either $R$ or $\ddot{C}$;
  so we return the current value of $\ddot{C}$.
\hfill \ebox
\end{example}

\section*{Appendix C: Partitions for \texttt{init\_rand}$_m$}
\label{appendix-randpart}
The large number of partitions required for the \texttt{init\_rand} family
is not necessarily obvious.
Assuming $n$ is non-negative, we must distinguish between the case where the array is empty
(thus $n = 0$) or non-empty ($n > 0$).
This gives us two base orderings: $[\{0, n\}]$, and $[\{0\},\{n\}]$.
In these descriptions, sets denote equivalence classes, and equivalence classes are
listed in increasing order.

If $n = 0$, there is only one possible value for $i_1$: $[\{0, i, n\}, \{i^+\}]$.
Otherwise, we must distinguish all the possible relations between $0$, $i$ and $n$.
The resulting orderings are as follows: 
\begin{center}
\begin{tabular}{c|c}
  $[\{0, i_1, n\}, \{i_1^+\}]$ & $[\{0\}, \{i_1, n\} \{i_1^+\}]$ \\ 
  \hline
  $[\{0, i_1\}, \{i_1^+, n\}]$ & $[\{0\}, \{i_1\}, \{i_1^+, n\}]$ \\
  \hline
  $[\{0, i_1\}, \{i_1^+\}, \{n\}]$ & $[\{0\}, \{i_1\}, \{i_1^+\}, \{n\}]$ \\
\end{tabular}
\end{center}
When we construct the partitions for $m = 2$, we introduce $i_2$ into all
feasible locations in each of the partitions for $m = 1$.

In cases where we must explicitly distinguish the $0^{th}$ element, this
progression grows substantially faster, with $m = 1, \ldots, 5$ yielding
$[9, 45, 333, 3285, 40509]$.
\end{document}